\documentclass[letterpaper]{article}
\usepackage[preprint]{aaai2027}
\usepackage[hyphens]{url}
\usepackage{graphicx}
\usepackage{natbib}
\usepackage{caption}
\usepackage{algorithm}
\usepackage{algorithmic}

\usepackage{booktabs}

\usepackage{amsmath}
\usepackage{amssymb}
\usepackage{mathtools}
\usepackage{amsthm}
\usepackage{thmtools}
\usepackage{thm-restate}
\usepackage[inline]{enumitem}
\usepackage[capitalize,noabbrev]{cleveref}
\usepackage{float}

\title{A Compositional Theory of Causally Masked Transformers}
\author{
    Franz Nowak\textsuperscript{\rm 1}, Ryan Cotterell\textsuperscript{\rm 1}, Reda Boumasmoud\textsuperscript{\rm 1}\\
    \normalfont\{franz.nowak, ryan.cotterell\}@inf.ethz.ch, reda.boumasmoud@math.ethz.ch 
}
\affiliations{
    \textsuperscript{\rm 1}ETH Zürich
}

\begin{document}
\renewcommand{\theequation}{\arabic{equation}}

\crefname{section}{Section}{Sections}
\Crefname{section}{Section}{Sections} 

\crefname{appendix}{Appendix}{Appendices}
\Crefname{appendix}{Appendix}{Appendices}
\crefformat{appendix}{Appendix~#2#1#3}
\Crefformat{appendix}{Appendix~#2#1#3}
\crefrangeformat{appendix}{Appendices~#3#1#4\textendash#5#2#6}
\Crefrangeformat{appendix}{Appendices~#3#1#4\textendash#5#2#6}

\crefname{table}{Table}{Tables}
\Crefname{table}{Table}{Tables}
\crefname{figure}{Fig.}{Figs.}
\Crefname{figure}{Figure}{Figures}

\crefname{algorithm}{Algorithm}{Algorithms}
\Crefname{algorithm}{Algorithm}{Algorithms}

\crefname{equation}{Eq.}{Eqs.}
\Crefname{equation}{Equation}{Equations}
\creflabelformat{equation}{#2#1#3} 

\crefname{theorem}{Theorem}{Theorems}
\Crefname{theorem}{Theorem}{Theorems}
\crefname{lemma}{Lemma}{Lemmas}
\Crefname{lemma}{Lemma}{Lemmas}
\crefname{corollary}{Corollary}{Corollaries}
\Crefname{corollary}{Corollary}{Corollaries}
\crefname{proposition}{Proposition}{Propositions}
\Crefname{proposition}{Proposition}{Propositions}
\crefname{definition}{Definition}{Definitions}
\Crefname{definition}{Definition}{Definitions}
\crefname{claim}{Claim}{Claims}
\Crefname{claim}{Claim}{Claims}
\crefname{example}{Example}{Examples}
\Crefname{example}{Example}{Examples}
\crefname{fact}{Fact}{Facts}
\Crefname{fact}{Fact}{Facts}
\crefname{question}{Question}{Questions}
\Crefname{question}{Question}{Questions}
\crefname{conjecture}{Conjecture}{Conjectures}
\Crefname{conjecture}{Conjecture}{Conjectures}
\crefname{remark}{Remark}{Remarks}
\Crefname{remark}{Remark}{Remarks}

\theoremstyle{plain}
\newtheorem{theorem}{Theorem}
\newtheorem{lemma}[theorem]{Lemma}
\newtheorem{proposition}[theorem]{Proposition}
\newtheorem{corollary}[theorem]{Corollary}
\newtheorem{claim}[theorem]{Claim}
\newtheorem{fact}[theorem]{Fact}
\newtheorem{definition}[theorem]{Definition}
\newtheorem{example}[theorem]{Example}
\newtheorem{question}[theorem]{Question}
\newtheorem{conjecture}[theorem]{Conjecture}
\newtheorem{condition}[theorem]{Condition}

\theoremstyle{remark}
\newtheorem{remark}[theorem]{Remark}

\newcommand{\R}{\mathbb{R}}
\newcommand{\C}{\mathbb{C}}
\newcommand{\Z}{\mathbb{Z}}
\newcommand{\F}{\mathbb{F}}
\newcommand{\Q}{\mathbb{Q}}
\newcommand{\N}{\mathbb{N}}

\makeatletter
\newcommand{\citeposs}[1]{%
  \citeauthor{#1}'s\ (\citeyear{#1})%
}
\makeatother

\newcommand{\defn}[1]{\textbf{#1}}
\newcommand{\defeq}{\mathrel{\stackrel{\textnormal{\tiny def}}{=}}}

\newcommand{\id}{\mathrm{Id}}
\newcommand{\sgn}{\mathrm{sign}}
\renewcommand{\ker}{\mathop{\mathrm{Ker}}}
\newcommand{\im}{\mathop{\mathrm{Im}}}
\newcommand{\ord}{\mathop{\text{ord}}}
\newcommand{\rank}{\mathop{\text{rank}}}
\newcommand{\Tr}{\mathop{\text{Tr}}}
\newcommand{\argmin}{\mathop{\text{argmin}}}
\newcommand{\argmax}{\mathop{\text{argmax}}}
\newcommand{\bigo}[1]{\mathcal{O}\left(#1\right)}

\newcommand{\norm}[1]{||#1||}
\newcommand{\innerprod}[1]{\langle{#1}\rangle}
\newcommand{\E}{\mathop{\mathbb{E}}}

\newcommand{\ug}{\leq}
\newcommand{\normal}{\trianglelefteq}
\newcommand{\teilt}{\big{|}}
\newcommand{\zykl}[1]{{<}{#1}{>}}
\newcommand{\ggt}{\mathop{\text{ggT}}}
\newcommand{\iso}{\simeq}
\newcommand{\isom}{\mathop{\text{Isom}}}
\newcommand{\aut}{\mathop{\text{Aut}}}
\newcommand{\stab}{\mathop{\text{Stab}}}
\newcommand{\operateson}{\curvearrowright}
\newcommand{\mat}{\mathrm{Mat}}
\newcommand{\trans}{\mathcal{T}}

\newcommand{\mA}{\mathbf{A}}
\newcommand{\mB}{\mathbf{B}}
\newcommand{\mC}{\mathbf{C}}
\newcommand{\mD}{\mathbf{D}}
\newcommand{\mlow}{\mathbf{M}^{\text{low}}}
\newcommand{\vu}{\boldsymbol{u}}
\newcommand{\vv}{\boldsymbol{v}}
\newcommand{\vx}{\boldsymbol{x}}

\newcommand{\rv}[1]{\underline{\underline{#1}}}
\newcommand{\noise}{\boldsymbol{w}}
\newcommand{\noiset}{\boldsymbol{\tilde{w}}}
\newcommand{\betahat}{\boldsymbol{\hat\beta}}
\newcommand{\xmax}{\tilde{x}}
\newcommand{\eps}{\varepsilon}
\newcommand{\Aff}{\mathop{\mathrm{Aff}}}
\newcommand{\I}{\text{I}}

\newcommand{\alphabet}{\Sigma}
\newcommand{\eosalphabet}{{\overline{\alphabet}}}
\newcommand{\posalphabet}{{\hat{\alphabet}}}
\newcommand{\kleene}[1]{{#1}^{\ast}}
\newcommand{\sym}{\sigma}
\newcommand{\str}{w}
\newcommand{\lang}{\mathcal{L}}
\newcommand{\emptyword}{\varepsilon}
\newcommand{\charfun}[1]{\chi_{#1}}

\newcommand{\monoid}{M}
\newcommand{\semigroup}{S}

\newcommand{\transsem}{T}

\newcommand{\wreathsem}{W}
\newcommand{\wreath}{\mathbb{W}}

\newcommand{\identity}[1]{\mathrm{id}_{#1}}

\newcommand{\subsem}{\leq}

\newcommand{\divides}{\prec}

\newcommand{\synsem}[1]{S(#1)}
\newcommand{\syncong}[1]{\sim_{#1}}

\newcommand{\genmon}[1]{\bigl\langle #1 \bigr\rangle_{\mathrm{mon}}}
\newcommand{\gensem}[1]{\bigl\langle #1 \bigr\rangle_{\mathrm{sem}}}

\newcommand{\wrclos}[1]{\langle #1 \rangle_{\wr}}

\newcommand{\pseudovariety}[1]{\mathbf{#1}}
\newcommand{\Aperiodic}{\pseudovariety{Ap}}
\newcommand{\Rtrivial}{\pseudovariety{R}}
\newcommand{\Definite}{\pseudovariety{Def}}
\newcommand{\LocallyR}{\pseudovariety{LR}}

\newcommand{\Rrel}{\mathcal{R}}
\newcommand{\Lrel}{\mathcal{L}}
\newcommand{\Hrel}{\mathcal{H}}
\newcommand{\Jrel}{\mathcal{J}}

\newcommand{\cyclic}[1]{\Z/{#1}\Z}

\newcommand{\transducer}{\mathcal{M}}

\newcommand{\extdelta}{\widehat{\delta}}
\newcommand{\extgamma}{\widehat{\gamma}}

\newcommand{\core}{\mathfrak{c}}
\newcommand{\acceptor}{\mathcal{A}}
\newcommand{\rnn}{\mathcal{R}}
\newcommand{\augrnn}{{\rnn^{+}}}
\newcommand{\transformer}{\mathcal{T}}
\newcommand{\augtransformer}{{\transformer^{+}}}
\newcommand{\acccore}{{\core_{\bullet}}}
\newcommand{\inicore}{{\core_{\circ}}}
\newcommand{\trivcore}{\core_{\mathrm{triv}}}

\newcommand{\cascadewith}[1]{\mathbin{\rhd_{#1}}}
\newcommand{\rnntype}[1]{\mathrm{type}(#1)}
\newcommand{\insettype}{\mathrm{In}}
\newcommand{\outsettype}{\mathrm{Out}}
\newcommand{\Algrnn}{\mathbf{AlgRNN}}

\newcommand{\inset}{X}
\newcommand{\insym}{x}
\newcommand{\inseq}{\mathbf{x}}
\newcommand{\outset}{Y}
\newcommand{\outsym}{y}
\newcommand{\outseq}{\mathbf{y}}
\newcommand{\states}{Z}
\newcommand{\state}{z}

\newcommand{\query}{{q}}
\newcommand{\Queries}{{Q}}
\newcommand{\qryop}{{\operatorname{qry}}}
\newcommand{\key}{{k}}
\newcommand{\Keys}{K}
\newcommand{\keyop}{{\operatorname{key}}}
\newcommand{\val}{{v}}
\newcommand{\Values}{V}
\newcommand{\valop}{{\operatorname{val}}}
\newcommand{\softmax}{\sigma}
\newcommand{\attention}{{\operatorname{Att}}}
\newcommand{\head}{H}
\newcommand{\prjfn}{\pi}
\newcommand{\mixfn}{\mu}
\newcommand{\similarity}{\alpha}

\newcommand{\hidsym}{\state}
\newcommand{\hidset}{\states}

\newcommand{\recfn}{f}
\newcommand{\outfn}{g}

\newcommand{\nlayers}{N}
\newcommand{\layer}{n}

\newcommand{\coretuple}{(\states, \inset, \outset, \recfn, \outfn)}
\newcommand{\corelayertuple}{(\states_\layer, \inset_\layer, \outset_\layer, \recfn_\layer, \outfn_\layer)}

\newcommand{\wiremap}{\tau}

\newcommand{\enc}{e}
\newcommand{\dec}{d}
\newcommand{\din}{{\dseq_\textit{in}}}
\newcommand{\dout}{{\dseq_\textit{out}}}
\newcommand{\encletter}{e_0}
\newcommand{\encword}{\widetilde{e}}

\newcommand{\initstate}{\state^{\circ}}
\newcommand{\accreg}{F}

\newcommand{\extF}[1]{\widehat{F}_{#1}}
\newcommand{\extG}[1]{\widehat{G}_{#1}}

\newcommand{\projlayer}[1]{\pi_{#1}}

\newcommand{\Lang}[1]{\lang(#1)}

\newcommand{\Phisem}{\Phi}

\newcommand{\hiddim}{m}
\newcommand{\seqdim}{d}
\newcommand{\seqlen}{T}
\newcommand{\tstep}{t}
\newcommand{\seqfn}{N}

\newcommand{\arith}{\mathfrak{M}}
\newcommand{\domain}{\mathcal{D}}
\newcommand{\ops}{\mathcal{O}}
\newcommand{\round}{\square}
\newcommand{\wrap}{\tikz{\node[circle,inner sep=0,outer sep=0,draw]{$\phantom{o}$}}}

\newcommand{\fp}{\textsc{fp}}
\newcommand{\integ}{\textsc{int}}

\newcommand{\tree}{\mathcal{T}}
\newcommand{\Trees}{\mathbf{Tree}}

\newcommand{\aperiodic}{\Aperiodic}

\newcommand{\countl}{\textsc{ModCount}}

\newcommand{\rmsnorm}{\mathrm{norm}}
\newcommand{\rms}{\mathrm{RMS}}
\newcommand{\relu}{\mathrm{ReLU}}
\newcommand{\diag}{\mathrm{Diag}}

\newcommand{\uone}{U_1}
\newcommand{\utwo}{U_2}
\newcommand{\uthree}{U_3}
\newcommand{\barUone}{\overline{U}_1}
\newcommand{\barUtwo}{\overline{U}_2}
\newcommand{\barUthree}{\overline{U}_3}

\newcommand{\realrnn}{\R\text{-}\mathbf{RNN}}
\newcommand{\finrnn}{\mathbf{Fin}\text{-}\mathbf{RNN}}
\newcommand{\functor}{\mathcal{F}}
\newcommand{\rnnsemigroup}{(\semigroup_\rnn,\states_\rnn)}
\newcommand{\layerwreath}{(\wreathsem_\rnn,\states_\rnn)}

\maketitle

\begin{abstract}
What types of decision problems can a causally masked, finite-precision transformer solve for inputs of arbitrary length? 
Existing answers often rely on idealized arithmetic, but under finite precision, rounding and evaluation order can change what information attention retains and therefore what the model can compute. 
We develop an algebraic formalization that derives expressivity directly from the model's implemented dynamics. 
Its central object is its memory; the finite internal state computed by attention that summarizes the information from the prefix available to all future queries. 
Each attention head updates its own state independently within a layer, while layers compose hierarchically, providing a uniform route from model assumptions to expressivity bounds.
Applying this method to transformers without positional embeddings, we obtain an expressivity hierarchy governed by the attention type under specific numerical semantics.
Width-one sliding-window attention supports bounded-suffix memory, while a modified form of soft attention supports irreversible, checklist-like state, and combining the two mechanisms provides an interplay of both.
Ordinary left-to-right floating-point soft attention can realize more expressive memory operations than any of the above.
Algebraically, the four cases correspond to definite, $\Rrel$-trivial, locally $\Rrel$-trivial, and aperiodic semigroups.
Under an explicit free-wiring assumption, all four bounds are tight.\looseness=-1
\end{abstract}

\section{Introduction}
We study the computational power of masked transformer language models whose precision and parameters are fixed independently of the input length.
Can they implement a reusable procedure, such as tracking latent state, enforcing a logical constraint, or counting, uniformly over arbitrarily long strings with fixed parameters? 
Formal language theory makes this question precise by asking which classes of languages a model family can recognize.
Existing transformer expressivity analyses typically answer it through an external formalism, such as linear temporal logic \citep{li2025characterizing}, circuit complexity \citep{merrill-etal-2022-saturated}, or restricted programming languages such as RASP \citep{pmlr-v139-weiss21a}; see \citet{strobl-etal-2024-formal} for a survey. 
Such characterizations are informative, but their analytical interface often changes with the modeling assumptions.\looseness=-1

Our approach differs in two related respects. First, we derive expressivity directly from the transition dynamics realized by the architecture, rather than through an external formalism. 
Second, we specify these dynamics at the level of the implemented arithmetic. 
Existing fixed-precision analyses model attention accumulation as associative and hence invariant to evaluation order. 
However, ordinary floating-point addition is not associative and can absorb small increments, so evaluation order can change both the induced state transitions and the information attention retains.

Under finite precision, causal masking, and a fixed left-to-right evaluation order, each head has a finite, implementation-sensitive state: a query-indexed collection of accumulators updated recursively as the prefix is read.
This finite state is a sufficient statistic of the prefix for the head's entire future behavior (\cref{prop:head-sufficiency}); it is its operational memory.
Unlike the growing key--value cache, it captures only information accessible through attention; unlike the current residual activation, it covers every possible future query rather than only the current one, placing it at the right level of abstraction when analyzing transformers as transition systems.

Attention is the mechanism that transforms this memory: each new key--value pair induces a transition of the head state.
The architecture composes these transformations, with heads acting independently in parallel and successive layers acting hierarchically.
These memory operations form transformation semigroups, making a causal transformer under our assumptions an algebraic RNN in the sense of \citet{nowak2026an}, so expressivity follows from the operations available to each attention mechanism and their composition across layers.\looseness=-1

We illustrate this analytic framework on finite-precision causal transformers without positional embeddings. 
At the level of individual heads, width-one sliding-window attention has a definite transition semigroup, sharp soft attention has an $\Rrel$-trivial transition semigroup, and full floating-point soft attention has an aperiodic transition semigroup.
Our composition results lift these head-level properties to upper bounds on complete transformers: architectures using width-one sliding-window heads recognize only definite languages, those using sharp soft-attention heads recognize only $\Rrel$-trivial languages, suitable cascades combining the two recognize only locally $\Rrel$-trivial languages, and those using full soft attention recognize only star-free languages.
The last bound excludes parity and all modular counting, irrespective of depth, width, or number of heads.
Under an explicit free-wiring assumption, all four bounds are tight.

The separation between sharp and full soft attention illustrates why arithmetic is part of the computational model.
Under floating-point evaluation, a contribution can be lost in the normalizer while still changing the accumulated value, creating memory behavior that the sharp variant excludes.
Algebraically, this is precisely the distinction between $\Rrel$-trivial and aperiodic behavior.

\begin{figure*}
    \centering
    \includegraphics[width=.6\linewidth, clip, trim=0cm 1.cm 0cm 1.7cm]{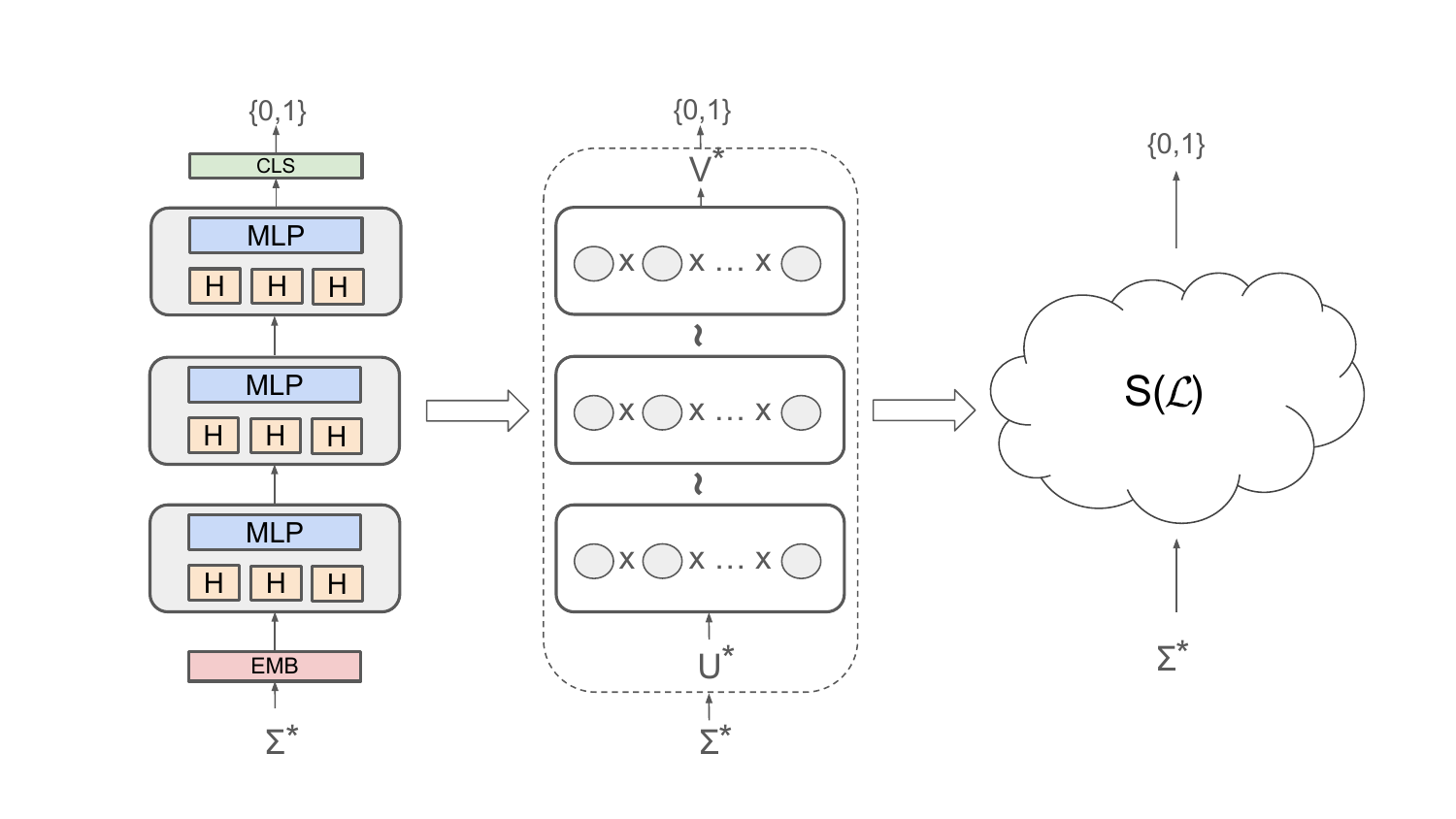}
    \caption{Abstraction of a transformer to an algebraic object to a syntactic semigroup.}
    \label{fig:abstraction}
\end{figure*}

\section{Algebraic Preliminaries}
Here we introduce the basic notions from semigroup theory, algebraic automata, and formal languages underlying this paper. Further preliminaries are found in \cref{app:prelim}.
For an accessible introduction, we refer the reader to \citet{pin:hal-03579131}. 

\begin{definition}
An \defn{alphabet} is a finite, nonempty set $\alphabet$, whose elements are called \defn{symbols}. A \defn{word} over $\alphabet$ is a finite sequence of symbols from $\alphabet$. 
\end{definition}
We write $\alphabet^+$ for the set of all non-empty finite words over $\alphabet$, and $\alphabet^*$ if this set additionally includes the empty word $\emptyword$. 

\begin{definition}
A \defn{language} over $\alphabet$ is a subset $\lang \subseteq \alphabet^*$. Its \defn{characteristic function} is the map
\begin{equation}
    \charfun{\lang} \colon \alphabet^* \to \{0,1\},
    \qquad
    \charfun{\lang}(\str) =
    \begin{cases}
    1, & \str \in \lang,\\
    0, & \str \notin \lang.
    \end{cases}
\end{equation}
\end{definition}

We treat transformer language models as language \emph{recognizers}. 
That is, the task is to compute $\charfun{\lang}$ rather than a distribution over the next token.

\begin{definition}
A \defn{semigroup} is a nonempty set $\semigroup$ equipped with an associative binary operation $\cdot$:
$
(s_1 \cdot s_2) \cdot s_3
=
s_1 \cdot (s_2 \cdot s_3)
$
for all $s_1,s_2,s_3 \in \semigroup$.
\end{definition}
When the operation is clear from context, we denote the semigroup simply by $\semigroup$.
For a given alphabet $\alphabet$, the set $\alphabet^+$ of nonempty words over $\alphabet$ is a semigroup under concatenation.
It is the free semigroup on $\alphabet$.

\begin{definition}
Let $\semigroup$ and $T$ be semigroups. A \defn{semigroup morphism} is a map $\phi \colon \semigroup \to T$ such that
$\phi(s_1 \cdot s_2) = \phi(s_1) \cdot \phi(s_2)$ for all $s_1,s_2 \in \semigroup$.
\end{definition}

In particular, a map from the alphabet $\alphabet$ into a semigroup $T$ determines a unique semigroup morphism from $\alphabet^+$ to $T$ by extending the map through concatenation.

\begin{definition}
A \defn{transformation semigroup} $(\semigroup,\states)$ consists of a state set $\states$ and a subsemigroup $\semigroup \subsem \states^{\states}$. More generally, a right action of a semigroup $\semigroup$ on $\states$ is a semigroup morphism
$\rho \colon \semigroup \to \states^{\states}$.
We write $\state \cdot s \defeq \rho(s)(\state)$. 
It follows that $\state\cdot(st) = (\state\cdot s)\cdot t$.
Replacing $\semigroup$ by its image $\rho(\semigroup)$ yields the corresponding transformation semigroup.
\end{definition}

We use two compositions of transformation semigroups. Parallel composition describes components whose state updates are independent, whereas cascade composition permits the transition of one component to depend on the state of another.\looseness=-1

\begin{definition}\label{def:parallel-composition}
Let $(\semigroup_1,\states_1),\ldots,(\semigroup_n,\states_n)$ be transformation semigroups. Their \defn{parallel composition} is the transformation semigroup
$
(\semigroup_1 \times \cdots \times \semigroup_n,
 \states_1 \times \cdots \times \states_n)
$
with componentwise action 
$(\state_1,\ldots,\state_n) \cdot (s_1,\ldots,s_n)
\defeq
(\state_1 \cdot s_1,\ldots,\state_n \cdot s_n)$.

\end{definition}

\begin{definition}\label{def:wreath-product}
Let $(\semigroup_1,\states_1)$ and $(\semigroup_2,\states_2)$ be transformation semigroups. Their \defn{left wreath product}, denoted
$
(\semigroup_1,\states_1) \wr (\semigroup_2,\states_2),
$
has state set $\states_1 \times \states_2$ and underlying set $\semigroup_1 \times \semigroup_2^{\states_1}$, where $\semigroup_2^{\states_1}$ is the set of maps $\states_1 \to \semigroup_2$. Its multiplication and action are given by
\begin{align}
(s_1,\phi_1)(s_2,\phi_2)
&\defeq
\bigl(s_1s_2,\; \state_1 \mapsto
\phi_1(\state_1)\phi_2(\state_1\cdot s_1)\bigr),\\
 (\state_1,\state_2) \cdot (s,\phi)
&\defeq
\bigl(\state_1 \cdot s, \state_2 \cdot \phi(\state_1) \bigr).
\end{align}
\end{definition}

The left-wreath convention evaluates $\phi$ at the first component's state before that component is updated. Iterated wreath products are associative up to canonical isomorphism, so we omit parentheses when no ambiguity arises.

\begin{definition}\label{def:division}
Let $\semigroup_1$ and $\semigroup_2$ be semigroups. We say that $\semigroup_1$ \defn{divides} $\semigroup_2$, written $\semigroup_1 \divides \semigroup_2$, if there exist a subsemigroup $T \subsem \semigroup_2$ and a surjective semigroup morphism
$\psi \colon T \twoheadrightarrow \semigroup_1.$
\end{definition}

Division is a preorder on semigroups: every semigroup divides itself, and if $\semigroup_1 \divides \semigroup_2$ and $\semigroup_2 \divides \semigroup_3$, then $\semigroup_1 \divides \semigroup_3$.

\section{Algebraic Transformers}
We algebraically describe a specific class of transformer language models that satisfy the following assumptions:
\begin{itemize}
    \item \textbf{F}: Finite-precision. All operations are carried out in constant-precision floating-point arithmetic.
    \item \textbf{M}: Causally masked attention. The attention score at position $t$ depends only on inputs $\leq t$.
    \item \textbf{P}: Positional embeddings are absent (NoPE).
    \item \textbf{E}: Fixed evaluation order. Attention is computed sequentially from left to right.
\end{itemize}
\subsection{Attention Heads}

\subsubsection{Attention Heads as Transition Systems}

We define the state of an attention head by expressing attention as a fold over finite accumulators.

\paragraph{Attention as a left-to-right fold.}
We first recall causal soft attention in its usual form.
At position $\tstep$, a head forms a query $\query_\tstep \in \Queries$, keys $\key_i \in \Keys$, and values $\val_i \in \Values$, where $1\leq i\leq t$, and $\Queries$, $\Keys$, and $\Values$ are the finite sets of possible queries, keys, and values (often in vector form), and returns
\begin{equation}
    \attention(\query_\tstep; \key_{< \tstep},\val_{< \tstep})
    =
    \sum_{i < \tstep}\frac{ \exp \similarity(\query_\tstep,\key_i)\,\val_i}
         {\sum_{j < \tstep} \exp \similarity(\query_\tstep,\key_j)},
\end{equation}
where $\similarity(\cdot,\cdot)$ is usually a scaled dot product when keys and queries are vectors.
Under assumption \textbf{E}, the sums are evaluated from left to right.
For a fixed query $\query$, define the soft-attention accumulator $a_\query = (n_\query,d_\query)$, representing the accumulated numerator and denominator, respectively.
Let $A$ denote the finite set of possible accumulator values.
To process a sequence of key-value pairs from left to right, we initialize the accumulator at the additive identity $a_{\query,0}=(0,0)\in A$.
Reading a new key-value pair $(\key_i,\val_i)$ at step $i\leq \tstep$ updates the running accumulator state $a_{\query,i}\in A$ via the transition function
\begin{equation}
    \delta_\query \colon A \times \Keys \times \Values \to A
\end{equation}
given by the recurrence
\begin{align}
    a_{\query,i} &= \delta_\query((n_{\query,i-1},d_{\query,i-1}),(\key_i,\val_i))\label{eq:att-rec}\\
    &=
    \bigl(n_{\query,i-1} + \exp \similarity(\query,\key_i)\val_i, d_{\query,i-1} + \exp \similarity(\query,\key_i)\bigr),\nonumber
\end{align}
with all operations evaluated under finite-precision semantics following assumption \textbf{F}.
The final attention output at step $\tstep$ is obtained only after completing the fold over all $i\leq t$, applying the readout
\begin{equation}
    \rho_\query \colon A \to \outset,
    \qquad
    \rho_\query(n,d) = n/d \label{eq:att-out}
\end{equation}
to the final accumulator state for query $\query$.
Other attention mechanisms can be formalized similarly by changing the accumulator and its update rule (see \cref{sec:case-study}).\footnote{Computing attention as a single query-dependent input pass is standard in online-attention algorithms like FlashAttention \citep{dao2022flashattention}, cf. online softmax \citep{Milakov2018OnlineNC}.}

\paragraph{The head state.}
We now pass from the attention implementation to the abstract state-transition system it induces.
Because the arithmetic domain is finite, the set $A$ of possible accumulator values is finite.
The relevant state is therefore not the unbounded list of past keys and values, but the finite summary produced by the prescribed scan.
The preceding accumulator was for a fixed query $\query$, but a head must be reusable for any query that may be produced later.
Since $\Queries$ is finite, we take \defn{the attention-head state set} to be the query-indexed family of accumulators
$
    \states_\head \defeq A^\Queries.
$

\paragraph{The head transitions.}
An element $\state \in \states_\head$ assigns to every possible query $\query \in \Queries$ the current prefix accumulator $\state(\query) \in A$.
Writing a new key-value pair $(\key,\val)$ acts on this state by updating every query accumulator:
\begin{equation}
    \Delta_{(\key,\val)}(\state)(\query)
    =
    \delta_\query(\state(\query),(\key,\val)),
    \qquad \query \in \Queries.
\end{equation}
The current query does not determine the state transition, only which component of the state is read via $\rho_\query(\state(\query))$.

\paragraph{The head transition system.}
Let the initial head state be $\initstate_\head \in \states_\head$.
If the current input of the head is $\insym \in \inset$ and the head computes query, key, and value maps $\query=\qryop(\insym)$, $\key=\keyop(\insym)$, and $\val=\valop(\insym)$, then the transition map is
\begin{equation}
    \recfn_\insym^\head \colon \states_\head \to \states_\head,
    \qquad
    \recfn_\insym^\head(\state)
    =
    \Delta_{(\keyop(\insym),\valop(\insym))}(\state),
\end{equation}
and the readout map is
\begin{equation}
    \outfn_\head \colon \states_\head \times \inset \to \outset,
    \qquad
    \outfn_\head(\state,\insym)
    =
    \rho_{\qryop(\insym)}
    \bigl(\state(\qryop(\insym))\bigr).
\end{equation}
Thus, the attention head induces a finite transition system
\begin{equation}
    (\states_\head,\initstate_\head,\inset,\recfn_\head),
    \qquad
    \recfn_\head(\state,\insym) = \recfn_\insym^\head(\state),
\end{equation}
with transition semigroup
\begin{equation}
    \transsem_\head
    \defeq
    \gensem{\recfn_\insym \mid \insym \in \inset}
    \subsem \states_\head^{\states_\head}.
\end{equation}
The map $\outfn_\head$ then specifies how this finite state is observed via the current query. 
Note it does not itself affect the transition semigroup of the head.

\begin{proposition}[Sufficiency of the head state]
\label{prop:head-sufficiency}
Let $\str_1,\str_2 \in \inset^{+}$ be prefixes reaching a common head state:
\begin{equation*}
\state \defeq \extdelta(\initstate_\head,\str_1) = \extdelta(\initstate_\head,\str_2).
\end{equation*}
Then for every $\str' \in \inset^{*}$ and every $\insym \in \inset$, the head produces the same output on the inputs $\str_1\str'\insym$ and $\str_2\str'\insym$.
\end{proposition}
The proposition says that the state $\extdelta(\initstate_\head,\str)$ is therefore a sufficient statistic of the prefix $\str$ for the entire future behavior of the head.
\begin{proof}
The extended transition satisfies the identity
\begin{equation}\label{eq:cocycle}
\extdelta(\state',uv) = \extdelta\bigl(\extdelta(\state',u),v\bigr),
\qquad \state' \in \states_\head,\ u,v \in \inset^{*}.
\end{equation}
The verification is by induction on $|v|$. 
By \cref{def:transducer}, the output of the head on the input $\str_i\str'\insym$ is
$\outfn_\head\bigl(\extdelta(\initstate_\head,\str_i\str'),\insym\bigr)$.
By \cref{eq:cocycle},
\begin{equation*}
\extdelta(\initstate_\head,\str_i\str')
= \extdelta\bigl(\extdelta(\initstate_\head,\str_i),\str'\bigr)
= \extdelta(\state,\str'),
\qquad i \in \{1,2\},
\end{equation*}
an expression in which $\str_i$ no longer occurs.
Both outputs therefore equal $\outfn_\head\bigl(\extdelta(\state,\str'),\insym\bigr)$.
\end{proof}
By the defining equation of $\Delta_{(\key,\val)}$, coordinate $\query$ of the new state depends on coordinate $\query$ of the old state alone: every generator is a coordinatewise transformation of $A^{\Queries}$.
Coordinatewise transformations form a subsemigroup of $\states_\head^{\states_\head}$, canonically isomorphic to $\prod_{\query\in\Queries} A^{A}$.
The following proposition records this passage; it reduces every pseudovariety question about the head to $|\Queries|$ questions about one-accumulator semigroups.

\begin{proposition}[Subdirect decomposition over queries]\label{prop:head-subdirect}
For $\query \in \Queries$, let
$\transsem_\query \defeq \gensem{\,\delta_\query(-,(\keyop(\insym),\valop(\insym))) \mid \insym \in \inset\,} \subsem A^{A}$.
The coordinatewise transformations of $\states_\head$ form a subsemigroup of $\states_\head^{\states_\head}$, canonically isomorphic to $\prod_{\query \in \Queries} A^{A}$ via $\bigl(\state \cdot (f_\query)_{\query}\bigr)(\query) \defeq \state(\query) \cdot f_\query $.
Under this identification,
\begin{equation*}
\transsem_\head \subsem \prod_{\query \in \Queries} \transsem_\query ,
\end{equation*}
with every projection surjective.
In particular, if $\mathbf{V}$ is a pseudovariety of finite semigroups and $\transsem_\query \in \mathbf{V}$ for every $\query$, then $\transsem_\head \in \mathbf{V}$.
\end{proposition}

\begin{proof}
The displayed map is an injective semigroup morphism onto the coordinatewise transformations: coordinatewise maps compose componentwise, and such a map determines its components by evaluation.
By definition of $\Delta_{(\key,\val)}$, each generator $\recfn^\head_\insym$ is coordinatewise, with $\query$-th component the generator $\delta_\query(-,(\keyop(\insym),\valop(\insym)))$ of $\transsem_\query$.
Hence $\transsem_\head$ lies in the image of $\prod_\query \transsem_\query$, projections carry generators onto generators, and $\mathbf{V}$ is closed under finite products and subsemigroups.
\end{proof}

\section{From Heads to Language Acceptors}
\begin{definition}[Transformer core]\label{def:transformer-core}
A \defn{transformer core} with heads $\head_1,\ldots,\head_m$ over a common input set $\inset$ is the tuple $\core=\coretuple$ with
\begin{align}
\states &\defeq \states_{\head_1}\times\cdots\times\states_{\head_m},\\
\recfn_\insym(\state_1,\ldots,\state_m) &\defeq \bigl(\recfn^{\head_1}_\insym(\state_1),\ldots,\recfn^{\head_m}_\insym(\state_m)\bigr),\\
\outfn\bigl((\state_1,\ldots,\state_m),\insym\bigr) &\defeq \mixfn\bigl(\outfn_{\head_1}(\state_1,\insym),\ldots,\outfn_{\head_m}(\state_m,\insym)\bigr),
\end{align}
where $\mixfn\colon \outset_1\times\cdots\times\outset_m\to\outset$ is a \defn{mixing map}, abstracting the output projection, normalization, and feedforward block.
\end{definition}
By \cref{def:parallel-composition}, the transition semigroup of a core is a subsemigroup of the parallel composition of its head semigroups,
$\semigroup_\core \defeq \gensem{\recfn_\insym \mid \insym\in\inset} \subsem \transsem_{\head_1}\times\cdots\times\transsem_{\head_m}$,
with every projection carrying generators onto generators; the mixing map contributes quotients of readouts and no transitions.

A transformer core is an algebraic core in the sense of \citet[Def.~3.1]{nowak2026an}: the sets are finite and both maps are total.
An \defn{algebraic transformer} $\transformer$ of depth $\nlayers$ is an algebraic RNN in the sense of \citet[\S3]{nowak2026an} all of whose cores are transformer cores. 
We keep the nomenclature of that paper unchanged: wiring maps, input and output types $\insettype$ and $\outsettype$, and cascade juxtaposition $\cascadewith{\wiremap}$ are used here without restatement.
An \defn{accepting core} $\acccore$ is a transformer core with output set $\{0,1\}$ whose readout depends on the state alone, $\outfn_\bullet(\state,\insym)=\charfun{\accreg}(\state)$ for an accepting region $\accreg\subseteq\states_\bullet$, intermediate readouts are Mealy, acceptance is Moore.
A \defn{transformer acceptor} is a tuple $\acceptor=(\alphabet,\augtransformer,\initstate,\enc)$ with $\augtransformer=\transformer\cascadewith{\wiremap_\nlayers}\acccore$, where $\transformer$ is an algebraic transformer, $\acccore$ is an accepting core, and with an initial global state $\initstate$, and an encoder $\enc\colon\alphabet\to\insettype(\transformer)$. It recognizes the language
\begin{equation}
\Lang{\acceptor}\defeq\{\str\in\alphabet^*\mid \projlayer{\bullet}(\extF{\augtransformer}(\initstate,\str))\in\accreg\},
\end{equation}
    where $\projlayer{\bullet}\colon \states_{\augtransformer}\to \states_\bullet$ projects to the accepting-core state and the empty word is acting as $\identity{\states_{\augtransformer}}$. Here, for a given word $\str=\str_0\cdots\str_{|\str|}\in\alphabet^*$, $\extF{\augtransformer}(\initstate,\str)$ denotes the global state obtained by iterating $\augtransformer$ on the encoding $\enc(\str_0)\cdots\enc(\str_{|\str|})$ starting from state $\initstate$.

Write $\wreathsem_{\augtransformer}$ for the ambient wreath product of the core transformation semigroups of $\augtransformer$, and $\wreathsem_{\augtransformer}^{T}$, $\semigroup_{\augtransformer}^{T}$ for the realized wreath semigroup and the realized transition semigroup relative to a set $T\subseteq\insettype(\transformer)$ of admissible inputs, defined through the layer-input dependency maps exactly as in \citet[\S3]{nowak2026an}.

\begin{proposition}[Factorization]\label{prop:transformer-factorization}
For every algebraic transformer $\transformer$ and every $T\subseteq\insettype(\transformer)$,
\begin{equation}
\semigroup_{\transformer}^{T}
\subsem
\wreathsem_{\transformer}^{T}
\subsem
\wreathsem_{\transformer}.
\end{equation}
\end{proposition}
\begin{proof}
This is the factorization lemma of \citet[Lem.~3.12]{nowak2026an}, specialized to transformer cores; the proofs transfer to semigroups by \cref{rem:semigroup-transfer}.
The first relation is an inclusion, not merely a division: the wreath product of faithful transformation semigroups is faithful.
\end{proof}

\begin{proposition}[Divisibility]\label{prop:transformer-acceptor-divisibility}
For every transformer acceptor $\acceptor=(\alphabet,\augtransformer,\initstate,\enc)$,
\begin{equation}
\synsem{\Lang{\acceptor}}
\divides
\semigroup_{\augtransformer}^{\enc(\alphabet)}
\subsem
\wreathsem_{\augtransformer}^{\enc(\alphabet)}
\subsem
\wreathsem_{\augtransformer}.
\end{equation}
\end{proposition}
\begin{proof}
By \cref{prop:eta-morphism}, $\str\mapsto\extF{\augtransformer}(-,\str)$ restricts to a semigroup morphism $\alphabet^{+}\to\semigroup_{\augtransformer}^{\enc(\alphabet)}$ recognizing $\Lang{\acceptor}\cap\alphabet^{+}$ through the accepting region, \cref{thm:syntactic-semigroup} gives the division, and \cref{prop:transformer-factorization} the two inclusions.
This is the acceptor divisibility chain of \citet{nowak2026an}, at semigroup level.
\end{proof}

\begin{remark}[Transfer of the proofs]\label{rem:semigroup-transfer}
The results of \citet[\S3]{nowak2026an} are stated for monoids and are used here for semigroups.
Their proofs transfer verbatim, for a structural reason: every object of {\em loc. cit.} section is defined by generators and closure, the generating sets are indexed by $\inset$ and never by $\inset\cup\{\emptyword\}$, and the identity occurs throughout as an inhabitant of a generated object, never as a hypothesis of an argument.
\end{remark}

\paragraph{A Note on Positional Embeddings}

With position embeddings, the encoder reads the position index alongside the symbol, $\enc \colon \alphabet \times \N \to \insettype(\transformer)$, and the set of realizable inputs remains finite under \textbf{F}.
What is lost is not finiteness but the morphism: the position-tagging map
\begin{equation}
\alphabet^+ \longrightarrow (\alphabet \times \N)^+,\,
\str \mapsto \bigl((\str_1,1),\ldots,(\str_{|\str|},|\str|)\bigr),
\end{equation}
does not commute with concatenation, since the same symbol at different positions has different images.
No morphism $\eta \colon \alphabet^+ \to \transsem$ is induced, \cref{def:semigroup-recognition} does not apply, and \cref{thm:syntactic-semigroup} has nothing to act on; in particular, the divisibility chain of \cref{prop:transformer-acceptor-divisibility} is not available.

\subsection{Expressivity Consequences}
The divisibility chain of \cref{prop:transformer-acceptor-divisibility} reduces transformer expressivity to a question of finite algebra.
We first fix a class of algebraic transformer acceptors by specifying the admissible attention mechanisms, finite-precision arithmetic, cores, wirings, encoders, and accepting cores.
\cref{prop:transformer-acceptor-divisibility} then shows that every recognized language has a syntactic semigroup dividing a realized wreath semigroup obtainable from the layer transition semigroups of the architecture.

This gives the following two-sided recipe, directly parallel to the algebraic RNN analysis of \citet{nowak2026an}:

\paragraph{Impossibility.}
For an impossibility result, it suffices to show that $\synsem{\lang}$ cannot divide any realized wreath semigroup generated by the permitted cores and wirings.
Equivalently, if a group, reset behavior, or other finite-semigroup component required by the syntactic semigroup of $\lang$ cannot occur in the wreath closure of the available core semigroups, then no acceptor in the class recognizes $\lang$.
Krohn--Rhodes theory \citep{krohn-rhodes-1965-algebraic}, and its extensions \citep{STIFFLER1973159}, supply one standard way to make this obstruction explicit, by decomposing finite transformation semigroups into simple interpretable basic components.

\paragraph{Construction witness.}
Conversely, to prove that a language is recognizable, one constructs transformer cores whose realized transition semigroups provide the necessary algebraic components, wires them so that the resulting realized wreath semigroup contains $\synsem{\lang}$ as a divisor, and chooses the accepting region $\accreg$ to separate the accepting syntactic classes.
The accepting core matters in this last step: the syntactic semigroup of $\Lang{\acceptor}$ need only divide the augmented semigroup $\semigroup_{\augtransformer}^{\enc(\alphabet)}$, not the transition semigroup of the base transformer alone.

For NoPE transformer acceptors, this yields a simple method to establish expressivity bounds: attention semantics determine the possible head semigroups, parallel heads determine core semigroups, and core cascades determine the realized wreath semigroups against which $\synsem{\lang}$ must be tested.\footnote{Position embeddings require separate treatment, because having access to unbounded position information allows a transformer to recognize non-regular languages just via its encoder.}\looseness=-1

\section{Expressivity of NoPE Transformers}\label{sec:case-study}
In this section, we derive concrete expressivity results for transformers without position embeddings (full specification in \cref{app:nope-architecture}).\looseness=-1

\subsection{Attention types}

We start by formally defining different types of attention similar to those used in practice.
The following attention mechanisms differ only in the accumulator update and readout in \cref{eq:streaming-attention-interface};
hence, each of them is a streaming attention mechanism in the sense fixed above.

\begin{definition}[Soft attention]\label{def:streaming-soft-attention}
    A head uses \defn{soft attention} if
    $a_{\layer,i}\in A$ consists of numerator--denominator pairs
    $(n_\query,d_\query)$, the initial accumulator is $a_{\layer,0}=(0,0)$, and its update and readout are given by \cref{eq:att-rec} and \cref{eq:att-out}, respectively.
    with $\similarity$ a (scaled) dot product, computed with fixed evaluation order under floating-point semantics.
    The readout at the empty prefix is fixed to be 0.
\end{definition}

\begin{definition}[Sharp soft attention]
\label{def:soft-attention-no-absorption}
    A soft-attention head is \defn{sharp} if, on every reachable accumulator update,
    either the denominator changes, or both the numerator and the denominator stay the same.
    I.e., if an attention score is small enough to not affect the normalizer, it is treated as zero.

    Writing $w=\exp(\similarity(\query,\key))$,
    $\widetilde n=\operatorname{fl}(n_\query+w\val)$, and
    $\widetilde d=\operatorname{fl}(d_\query+w)$, its update is
    \begin{equation}
        \delta_{\layer,i,\query}
        ((n_\query,d_\query),(\key,\val))
        =
        \begin{cases}
            (\widetilde n,\widetilde d),
            &\text{if }\widetilde d>d_\query,\\
            (n_\query,d_\query),&\text{otherwise}.
        \end{cases}
    \end{equation}
\end{definition}

\begin{definition}[Sliding-window attention of width one]
\label{def:streaming-window-one-attention}
    Fix an initial accumulator
    $a^\circ_{\layer,i,\query}
    =(\key^\circ_{\layer,i},\val^\circ_{\layer,i})
    \in\Keys_{\layer,i}\times\Values_{\layer,i}$.
    A head uses \defn{sliding-window attention of width one} if
    $A_{\layer,i}=\Keys_{\layer,i}\times\Values_{\layer,i}$ and
    \begin{equation}
        \delta_{\layer,i,\query}(a,(\key,\val))=(\key,\val),
        \qquad
        \rho_{\layer,i,\query}(\key,\val)=\val.
    \end{equation}
    Thus, at position $\tstep$, the head reads only the value stored at
    position $\tstep-1$, with $\val^\circ_{\layer,i}$ read at the first
    position.  
\end{definition}

\subsection{Available Semigroups by Attention Type}
\begin{definition}\label{def:definite}
    A semigroup $\semigroup$ is \defn{definite}, $\semigroup\in\Definite$, if
    there exists $n \geq 1$ such that $st=t$ for all
    $s\in\semigroup$ and $t\in\semigroup^n$.
\end{definition}
\begin{definition}\label{def:rtrivial}
    A semigroup $\semigroup$ is \defn{$\Rrel$-trivial},
    $\semigroup\in\Rtrivial$, if
    $s\semigroup^1=t\semigroup^1$ implies $s=t$ for all
    $s,t\in\semigroup$.
\end{definition}
\begin{definition}\label{def:locrtrivial}
    A semigroup $\semigroup$ is \defn{locally $\Rrel$-trivial},
    $\semigroup\in\LocallyR$, if $e\semigroup e$ is
    $\Rrel$-trivial for every idempotent $e\in\semigroup$.
\end{definition}
\begin{definition}\label{def:aperiodic}
    A semigroup $\semigroup$ is \defn{aperiodic},
    $\semigroup\in\Aperiodic$, if for every $s\in\semigroup$ there
    exists $n\geq 1$ such that $s^n=s^{n+1}$.
\end{definition}

\begin{figure}[ht]
    \centering
    \includegraphics[width=.85\linewidth, clip, trim=1.cm 2.2cm 8cm 0cm]{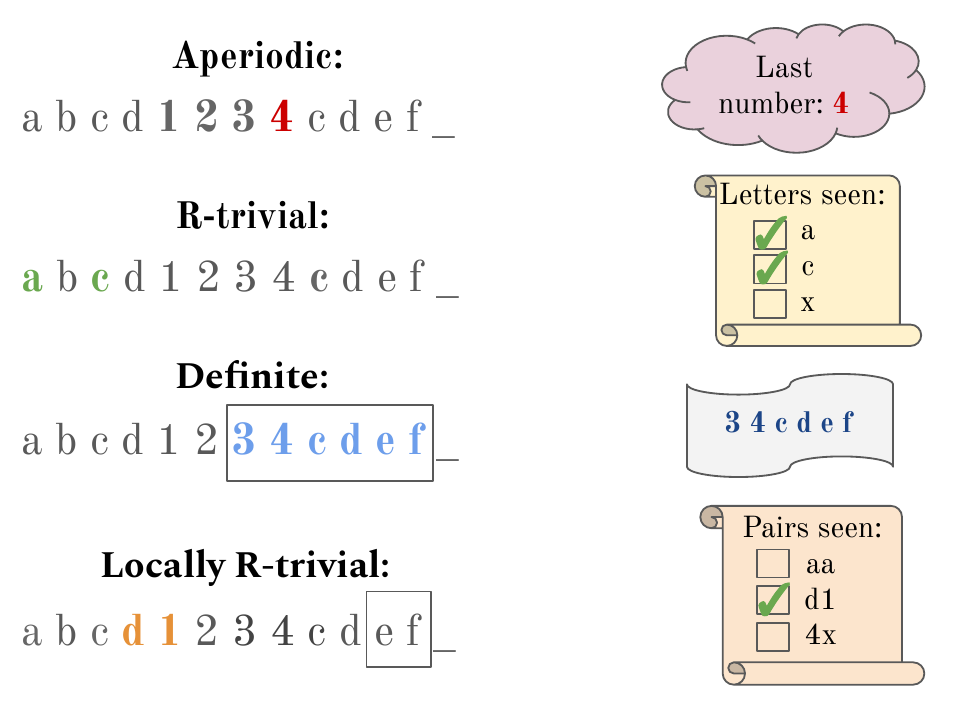}
    \caption{Illustration of the types of tasks that attention can perform depending on the variety of its transition semigroup.}
    \label{fig:semigroups}
\end{figure}

For intuition about the kinds of memory supported by these semigroup varieties, see \cref{fig:semigroups}. 
Aperiodic transition semigroups may retain finite-state information from arbitrarily far back, including ordering information; for example, they can record the last number seen.
In an $\Rrel$-trivial transition semigroup, state transformations progress without nontrivial reversible cycles, giving memory that behaves like a checklist rather than a rewritable register. 
Definite semigroups capture bounded-suffix memory, such as retaining a fixed-size window of the input.
Combining definite and $\Rrel$-trivial behavior yields locally $\Rrel$-trivial computations, which can retain checklist information about patterns detected within bounded windows.\looseness=-1

\begin{remark}
    Any semigroup from the classes described above can be decomposed into a wreath product whose factors come from just three prime semigroups, called $\uone,\utwo,\uthree$ (see \cref{app:prime-semigroups} for their definitions). The relevant decomposition results due to \citet{STIFFLER1973159} are outlined in \cref{app:subregular-kr}.
\end{remark}

We use the following lemmas (full proofs in \cref{app:prime-semigroups}).
\begin{restatable}{lemma}{windowOneDefinite}
\label{lem:window-one-definite}
    Every width-one sliding-window attention head has a definite
    transition semigroup, i.e., $\transsem_\head\in\Definite$.
\end{restatable}
\begin{proof}[proof sketch]
    Every nonempty input word acts as the constant map determined by its final key--value pair; hence, $st=t$, implying \cref{def:definite}.
\end{proof}

 For every ${i \in \{1,2,3\}}$, we define the transformation semigroup $\overline{U}_i \defeq (U_i, U_i^1)$ where $U_i$ (see \cref{def:theUs}) acts on $U_i^1$ from the right.
\begin{restatable}{lemma}{windowOneContainsUone}
\label{lem:window-one-contains-uone}
    There exists a width-one sliding-window attention head $\head$ such
    that $\barUone \divides (\transsem_\head,\states_\head)$.
\end{restatable}
\begin{proof}[proof sketch]
    On the invariant set consisting of the initial accumulator and two distinct key--value states, the corresponding inputs act as the two constant maps of $\barUone$.
\end{proof}

\begin{restatable}{lemma}{sharpSoftRtrivial}
\label{lem:sharp-soft-r-trivial}
    Every sharp soft-attention head has an $\Rrel$-trivial transition
    semigroup, i.e., $\transsem_\head\in\Rtrivial$.
\end{restatable}
\begin{proof}[proof sketch]
    Denominators weakly increase, and sharpness makes any transition that leaves all denominators fixed a self-loop. Thus reachability is a partial order, so the transition semigroup is $\Rrel$-trivial.
\end{proof}

\begin{restatable}{lemma}{sharpSoftContainsUtwo}
\label{lem:sharp-soft-contains-utwo}
    There exists a sharp soft-attention head $\head$ such that
    $\barUtwo\divides(\transsem_\head,\states_\head)$.
\end{restatable}
\begin{proof}[proof sketch]
    On an invariant two-state set straddling an absorption threshold, an underflowing input acts as the identity and a unit-weight input as the constant map to the saturated state, realizing $\barUtwo$.
\end{proof}

\begin{restatable}{lemma}{softAttentionAperiodic}
\label{lem:soft-attention-aperiodic}
    Every finite-precision soft-attention head has an aperiodic transition semigroup.
\end{restatable}
\begin{proof}[proof sketch]
    Under repetition of any input word, each denominator weakly increases and eventually stabilizes; thereafter each numerator coordinate evolves by an order-preserving map of a finite chain and also stabilizes. Hence every transition has an idempotent power.
\end{proof}

\begin{restatable}{lemma}{softAttentionContainsUthree}
\label{lem:soft-attention-contains-uthree}
    There exists a soft-attention head $\head$ such that $\barUthree \divides (\transsem_\head,\states_\head)$, i.e., its transition semigroup is \emph{not} $\Rrel$-trivial.\looseness=-1
\end{restatable}
\begin{proof}[proof sketch]
    On a saturated-denominator invariant slice, left-to-right rounded numerator updates realize the identity and two distinct constant maps, and hence $\barUthree$.
\end{proof}

Throughout this section, the decoder of
\cref{def:nope-transformer} is memoryless.  Equivalently, it may be
represented by an accepting core with trivial transition semigroup, so
appending it does not change any of the semigroup classes below.

\begin{theorem}\label{cor:window-definite}
    Any language accepted by a NoPE transformer with width-one
    sliding-window attention is definite.
\end{theorem}

\begin{proof}
    Every head transition semigroup is definite by
    \cref{lem:window-one-definite}.  Finite direct products and
    subsemigroups preserve definiteness, so every core transition
    semigroup is definite.  By \cref{prop:transformer-factorization},
    the realized transformer semigroup divides their cascade and is
    definite by \cref{thm:definite-decomposition}.  Its recognized
    language has a syntactic semigroup dividing this realized semigroup
    by \cref{prop:transformer-acceptor-divisibility}; closure under
    division therefore makes the language definite.
\end{proof}

\begin{theorem}\label{cor:sharp-r-trivial}
    Any language accepted by a NoPE transformer with sharp soft
    attention is $\Rrel$-trivial.
\end{theorem}

\begin{proof}
    By \cref{lem:sharp-soft-r-trivial}, every head transition semigroup
    is $\Rrel$-trivial.  Finite direct products and subsemigroups of
    $\Rrel$-trivial semigroups are $\Rrel$-trivial, so every core
    transition semigroup is $\Rrel$-trivial.  The realized transformer
    semigroup divides the cascade of these core semigroups by
    \cref{prop:transformer-factorization}; hence it is
    $\Rrel$-trivial by \cref{thm:rtrivial-decomposition}.  Finally,
    \cref{prop:transformer-acceptor-divisibility} and closure under
    division imply that the syntactic semigroup of the accepted language
    is $\Rrel$-trivial.
\end{proof}
The above theorem recovers by an independent algebraic route a similar result obtained through linear temporal logic by \citet{li2025characterizing}. 

\begin{theorem}\label{cor:mixed-locally-r-trivial}
    Any language accepted by a transformer whose lower cores have sliding-window attention and whose higher cores have sharp attention heads, is locally $\Rrel$-trivial.
\end{theorem}

\begin{proof}
    Take a transformer each of whose lower cores have width-one
    sliding-window attention heads and whose higher cores sharp soft-attention heads. 
    For the lower heads, their transition semigroups are definite by \cref{lem:window-one-definite}; for the higher ones, they are
    $\Rrel$-trivial by \cref{lem:sharp-soft-r-trivial}.  
    The factorization of
    \cref{prop:transformer-factorization} places the realized
    transformer semigroup below a cascade of these core semigroups.
    By \cref{thm:locally-r-decomposition}, it is locally
    $\Rrel$-trivial.  The syntactic semigroup of the accepted language
    divides the realized semigroup by
    \cref{prop:transformer-acceptor-divisibility}, so it is also locally
    $\Rrel$-trivial.
\end{proof}

\begin{theorem}\label{cor:soft-star-free}
    Any language accepted by a NoPE transformer with soft attention is star-free.
\end{theorem}

\begin{proof}
    By \cref{lem:soft-attention-aperiodic}, the transition semigroup of
    every attention head is aperiodic.  The transition semigroup of each
    transformer core is a subsemigroup of the finite direct product of
    its head semigroups, and is therefore aperiodic.  By
    \cref{prop:transformer-factorization}, the realized transition
    semigroup of the transformer divides the cascade of its core
    semigroups.  Aperiodic semigroups are closed under finite direct
    products, wreath products, and division; equivalently, by
    \cref{thm:aperiodic-decomposition}, they are precisely the
    semigroups dividing finite cascades of copies of $\uthree$.
    Consequently, the realized transformer semigroup is aperiodic.
    The encoder restricts the admissible inputs but cannot enlarge this
    semigroup, and the memoryless decoder adds no transitions.  By
    \cref{prop:transformer-acceptor-divisibility}, the syntactic
    semigroup of the accepted language divides the realized transformer
    semigroup and is therefore aperiodic.  The algebraic
    characterization of star-free languages as those with an aperiodic syntactic semigroup \citep{schutzenberger1955theorie} yields the claim.\looseness=-1
\end{proof}

\begin{definition}[Free-wiring assumption]\label{ass:free-wiring}
Say a class of NoPE transformers satisfies the \defn{free-wiring assumption} if the encoders, the per-head key, query, and value maps, the mixing maps, the wiring maps, and the accepting regions may be arbitrary total maps between the corresponding finite sets, and the cascade may place cores with the required attention types in the order prescribed by \cref{thm:definite-decomposition,thm:rtrivial-decomposition,thm:locally-r-decomposition,thm:aperiodic-decomposition}.
\end{definition}

\begin{corollary}\label{cor:attention-bounds-tight}
    Under the free-wiring assumption, the above bounds are tight.
\end{corollary}

\begin{proof}
    All the lower-bound components required by the decomposition theorems are available within the corresponding attention mechanisms:
    width-one sliding-window attention realizes
    $\barUone$ by \cref{lem:window-one-contains-uone}, sharp soft
    attention realizes $\barUtwo$ by
    \cref{lem:sharp-soft-contains-utwo}, and floating-point soft
    attention realizes $\barUthree$ by
    \cref{lem:soft-attention-contains-uthree}.  The decomposition theorems,
    \cref{thm:definite-decomposition,thm:rtrivial-decomposition,thm:locally-r-decomposition,thm:aperiodic-decomposition},
    express every semigroup in the respective upper-bound class as a
    divisor of a finite cascade of these components.  If the transformer
    wirings realize the required cascade controls and the decoder
    realizes the required accepting subset, the resulting transformer
    recognizes every language in the corresponding class. 
    Together with the preceding upper bounds, this proves the tightness under the
    stated implementability condition.\looseness=-1
\end{proof}

\begin{table}[H]
    \centering
    \begin{tabular}{c|c c c}
        Attention & Primes & Variety & Languages \\
        \hline
        sliding & $\uone$ & $\Definite$ & definite \\
        sharp & $\utwo$ & $\Rtrivial$ & $\Rrel$-trivial \\
        sliding + sharp & $\uone,\utwo$ & $\LocallyR$ & loc. $\Rrel$-trivial \\
        soft & $\uone, \utwo, \uthree$ & $\aperiodic$ & star-free
    \end{tabular}
    \caption{Result overview. We map the attention types to the realizable prime semigroups, the resulting semigroup varieties closed under direct product, wreath product, and division, and the corresponding languages that NoPE transformers with these attention types can recognize. Note that $\uone$ and $\utwo$ are subsemigroups of $\uthree$.}
    \label{tab:results}
\end{table}

\subsection{The Importance of Evaluation Order}
Assumption \textbf{E} aligns temporal order with the right action of the transformation semigroup.
Reversing the evaluation order would exchange right zeros with left zeros and replace the $\Rrel$-trivial and locally $\Rrel$-trivial conclusions by their dual $\Lrel$-trivial and locally $\Lrel$-trivial conclusions.

\section{Related Work}
Our framework extends \citeposs{nowak2026an} algebraic analysis of RNN language models.
We identify assumptions under which NoPE transformers become algebraic RNNs, then use refinements of Krohn--Rhodes theory to classify different attention mechanisms and derive compositional expressivity bounds.

Closest are formal-language characterizations of masked transformers.
Using Krohn--Rhodes theory, \citet{yang2024masked} show that fixed-precision transformers with unique hard attention recognize exactly the star-free languages.
Using LTL, \citet{li2025characterizing} obtain an $\Rrel$-trivial characterization for fixed-precision soft attention.
We recover their bound for sharp soft attention but show that it depends on the numerical semantics: left-to-right floating-point soft attention yields the broader aperiodic bound.
Concurrently, \citet{li-cotterell-2026-characterizing} use LTL to show that local attention yields definite languages and that combining it with idealized soft attention yields locally $\Rrel$-trivial languages, matching our results at the level of language classes.

Using Krohn--Rhodes theory, \citet{liu2023transformers} obtain solvable-language expressivity under growing precision and length-dependent wiring; we instead assume fixed precision and fixed parameters.

\section{Discussion}
This work establishes an algebraic framework for the underlying dynamics of masked finite-precision transformers under explicit numerical semantics, by showing that they can be formalized as RNNs.
This allows placing lower and upper bounds on the languages such architectures can recognize.

The analysis is restricted to transformers without positional embeddings, because mapping from strings to position-indexed strings is not a morphism.
This restriction is substantive: an encoder with access to unbounded position information need not induce the fixed finite input alphabet used by the semigroup analysis, and may support nonregular behavior independently of the attention mechanism. 
While unbounded position information cannot be assessed using this algebraic framework, certain finite or repeating position embeddings could be incorporated in future work.

The query-indexed accumulator state can be extremely large. By \cref{prop:head-subdirect}, however, every pseudovariety question about a head reduces to the one-accumulator semigroups $\transsem_\query$, so the size of $A^{\Queries}$ is harmless for the bounds. 
The algebraic description is therefore primarily a structural tool for proving expressivity bounds, rather than a prescription for explicitly enumerating the transition semigroup of a trained model.
The dependence of the results on numerical semantics is nevertheless practically relevant for heavily quantized architectures, where the sets of queries and accumulator values are substantially reduced \citep{Duan2023BitformerAE}.  
More broadly, the decomposition into head, core, and cascade semigroups suggests a route to architecture design: attention mechanisms can be chosen to supply particular algebraic components, while wiring maps determine how these components interact across layers. 
This perspective may enable attention variants whose expressivity is controlled by construction.\looseness=-1

\bibliography{references}

\clearpage
\appendix
\crefalias{section}{appendix}

\section{Further Preliminaries}\label{app:prelim}
\subsection{Recognition of Languages by Semigroups}

\begin{definition}\label{def:semigroup-recognition}
A semigroup $\semigroup$ \defn{recognizes} a language $\lang \subseteq \alphabet^+$ if there exist a semigroup morphism $\eta \colon \alphabet^+ \to \semigroup$ and a subset $P \subseteq \semigroup$ such that $\lang = \eta^{-1}(P)$.
\end{definition}
For languages $\lang \subseteq \alphabet^*$, the empty word is handled separately by whether $\emptyword \in \lang$; the semigroup recognizes the nonempty part $\lang \cap \alphabet^+$.
This is analogous to recognition by the syntactic monoid, except that the semigroup presentation does not require the empty word to be part of the input semigroup.
This choice loses no information for language recognition: any semigroup $\semigroup$ can be made into a monoid by adjoining an identity element, and the free monoid over $\alphabet$ is $\alphabet^* = \alphabet^+ \cup \{\emptyword\}$.

\begin{definition}\label{def:syntactic-semigroup}
For a language $\lang \subseteq \alphabet^*$, the \defn{syntactic congruence} $\syncong{\lang}$ on $\alphabet^+$ is defined by
\begin{equation}
u \syncong{\lang} v
\quad\Longleftrightarrow\quad
\text{for all } x,y \in \alphabet^*,\;
xuy \in \lang \Longleftrightarrow xvy \in \lang.
\end{equation}
The quotient semigroup $\synsem{\lang} \defeq \alphabet^+/{\syncong{\lang}}$ is the \defn{syntactic semigroup} of $\lang$.
\end{definition}

\begin{theorem}\label{thm:syntactic-semigroup}
If a semigroup $\semigroup$ recognizes $\lang \cap \alphabet^+$, then $\synsem{\lang} \divides \semigroup$. Moreover, $\lang$ is regular if and only if $\synsem{\lang}$ is finite.
\end{theorem}

Thus, the syntactic semigroup is minimal up to division among the semigroups recognizing the nontrivial part of a language. It provides the algebraic interface between language recognition and the transition semigroups of concrete state machines.

\subsection{Transducers and Acceptors}

\begin{definition}\label{def:transducer}
A \defn{transducer} is a tuple
\begin{equation}
\transducer = (\states,\initstate,\inset,\outset,\delta,\gamma),
\end{equation}
where $\states$ is a state set with initial state $\initstate \in \states$, $\inset$ and $\outset$ are input and output sets, $\delta \colon \states \times \inset \to \states$ is a transition map, and $\gamma \colon \states \times \inset \to \outset$ is an output map (Mealy form \citep{mealy1955}; the Moore form is the special case in which $\gamma$ does not depend on its second argument \citep{moore1956}).
\end{definition}

The transition map extends recursively to words as $\extdelta \colon \states \times \inset^* \to \states$, where
\begin{equation}
\extdelta(\state,\emptyword) = \state,
\qquad
\extdelta(\state,\str\insym)
= \delta(\extdelta(\state,\str),\insym).
\end{equation}
Thus, on input $\str \in \inset^*$, the transducer produces the output $\gamma(\extdelta(\initstate,\str))$.

\begin{definition}\label{def:acceptor}
An \defn{acceptor} over $\alphabet$ is a tuple
\begin{equation}
\acceptor = (\states,\initstate,\alphabet,\delta,\accreg),
\end{equation}
where $\states$, $\initstate$, and $\delta$ are as in \cref{def:transducer}, and $\accreg \subseteq \states$ is the set of accepting states. The language recognized by $\acceptor$ is
\begin{equation}
\Lang{\acceptor}
\defeq
\{\str \in \alphabet^* \mid \extdelta(\initstate,\str) \in \accreg\}.
\end{equation}
Equivalently, an acceptor is a transducer with output set $\{0,1\}$ and output map $\charfun{\accreg}$.
\end{definition}

\subsection{Transition Semigroups of Transducers}

\begin{definition}\label{def:transition-semigroup}
Let $\transducer$ be a transducer. Each input $\insym \in \inset$ induces a state transformation
\begin{equation}
\delta_{\insym} \colon \states \to \states,
\qquad
\delta_{\insym}(\state) \defeq \delta(\state,\insym).
\end{equation}
The \defn{transition semigroup} of $\transducer$ is the transformation semigroup
\begin{equation}
\transsem_{\transducer}
\defeq
\gensem{\delta_{\insym} \mid \insym \in \inset}
\subsem \states^{\states}.
\end{equation}
\end{definition}
The elements of $\transsem_{\transducer}$ are precisely the state transformations induced by nonempty input words. The empty word acts as the identity transformation on $\states$.

\begin{lemma}\label{prop:eta-morphism}
The map $\eta \colon \inset^+ \to \transsem_\transducer$ defined by $\state \cdot \eta(\str) \defeq \extdelta(\state,\str)$ is a semigroup morphism.
\end{lemma}
\begin{proof} Follows readily by induction on $|\str|$.
\end{proof}

\begin{remark}[Why semigroups]\label{rem:why-semigroups}
The subregular classes of the case study are invisible at monoid level, so the passage to semigroups is necessary, not stylistic.
A definite monoid is trivial: if $1 \in \semigroup^{n}$, then $s \cdot 1 = 1$ for every $s$.
Likewise $1\semigroup 1 = \semigroup$, so on monoids $\LocallyR$ collapses onto $\Rtrivial$.
Two of the four classes below therefore have no monoid-level content; the semigroup of a width-one window head is definite and ceases to be so the moment an identity is adjoined.
\end{remark}

\section{NoPE Transformer Architecture}\label{app:nope-architecture}

We now instantiate the algebraic transformer of the preceding section as a finite-precision transformer with no position embeddings (NoPE transformer).
Following the streaming formulation of attention above, a head reads the strict prefix state stored in its current state and then incorporates the key--value pair at the current position.

Fix a finite arithmetic domain $\domain$, a depth $\nlayers\in\Z^+$, a model dimension $\seqdim\in\Z^+$, and numbers of heads $m_1,\ldots,m_\nlayers$.  
For every core $\layer\in\{1,\ldots,\nlayers\}$, set
\begin{equation}
    \inset_\layer=\outset_\layer=\domain^\seqdim.
\end{equation}
Head $i$ of layer $\layer$ has finite query, key, value, and head-output
sets $\Queries_{\layer,i}$, $\Keys_{\layer,i}$,
$\Values_{\layer,i}$, and $\outset_{\layer,i}$, together with learned maps
\begin{align}
    \qryop_{\layer,i}&\colon\inset_\layer\to\Queries_{\layer,i},
    &
    \keyop_{\layer,i}&\colon\inset_\layer\to\Keys_{\layer,i},
    \\
    \valop_{\layer,i}&\colon\inset_\layer\to\Values_{\layer,i}.
\end{align}
Its attention mechanism is specified by a finite accumulator set
$A_{\layer,i}$, initial accumulators
$a^\circ_{\layer,i,\query}\in A_{\layer,i}$, and maps
\begin{equation}
    \delta_{\layer,i,\query}\colon
    A_{\layer,i}\times\Keys_{\layer,i}\times\Values_{\layer,i}
    \to A_{\layer,i},
    \qquad
    \rho_{\layer,i,\query}\colon A_{\layer,i}\to\outset_{\layer,i},
    \label{eq:streaming-attention-interface}
\end{equation}
for every $\query\in\Queries_{\layer,i}$.  Thus,
\begin{equation}
    \states_{\layer,i}=A_{\layer,i}^{\Queries_{\layer,i}},
    \qquad
    \state^\circ_{\layer,i}(\query)=a^\circ_{\layer,i,\query}.
\end{equation}

The transformer core
$\core_\layer=(\states_\layer,\inset_\layer,\outset_\layer,
\recfn_\layer,\outfn_\layer)$ has the product state
\begin{equation}
    \states_\layer
    =\prod_{i=1}^{m_\layer}\states_{\layer,i}.
\end{equation}
For $\state=(\state_1,\ldots,\state_{m_\layer})$ and
$\insym\in\inset_\layer$, its recurrence is defined componentwise by
\begin{equation}
    \recfn_\layer(\state,\insym)_i(\query)
    =
    \delta_{\layer,i,\query}\!\left(
        \state_i(\query),
        \keyop_{\layer,i}(\insym),
        \valop_{\layer,i}(\insym)
    \right).
    \label{eq:nope-core-recurrence}
\end{equation}
The head outputs and core readout are
\begin{align}
    h_{\layer,i}(\state_i,\insym)
    &=
    \rho_{\layer,i,\qryop_{\layer,i}(\insym)}
    \!\left(\state_i(\qryop_{\layer,i}(\insym))\right),\\
    \outfn_\layer(\state,\insym)
    &=
    W^O_\layer\!\left(
        h_{\layer,1}(\state_1,\insym),\ldots,
        h_{\layer,m_\layer}(\state_{m_\layer},\insym)
    \right),
    \label{eq:nope-core-readout}
\end{align}
where $W^O_\layer$ is the finite-precision output projection.  Consequently,
at position $\tstep$, the core computes
\begin{equation}
    \state_{\layer,\tstep}
    =\recfn_\layer(\state_{\layer,\tstep-1},\insym_{\layer,\tstep}),
    \qquad
    \outsym_{\layer,\tstep}
    =\outfn_\layer(\state_{\layer,\tstep-1},\insym_{\layer,\tstep}).
\end{equation}
In particular, the readout uses precisely the strict prefix accumulated
before position $\tstep$.

Let $\operatorname{LN}_\layer\colon\domain^\seqdim\to\domain^\seqdim$
be layer normalization with learned scale and bias, and let
$\operatorname{MLP}_\layer\colon\domain^\seqdim\to\domain^\seqdim$
be a finite-precision feedforward network.  The wiring from layer $\layer$
to layer $\layer+1$ is
\begin{align}
    r_\layer(\insym,\outsym)
    &=\operatorname{LN}_\layer^{\mathrm{att}}(\insym+\outsym),\\
    \wiremap_\layer(\insym,\outsym)
    &=\operatorname{LN}_\layer^{\mathrm{mlp}}\!\left(
        r_\layer(\insym,\outsym)
        +\operatorname{MLP}_\layer(r_\layer(\insym,\outsym))
    \right),
    \label{eq:nope-wiring}
\end{align}
with every operation evaluated in the fixed finite-precision semantics.
This wiring includes the attention residual connection, layer
normalization, the MLP, and the MLP residual connection.

The encoder and decoder are set maps
\begin{equation}
    \enc\colon\alphabet\to\domain^\seqdim,
    \qquad
    \dec\colon\domain^\seqdim\to\{0,1\}.
\end{equation}
The encoder factors through the one-hot embedding $\alphabet\hookrightarrow\domain^{|\alphabet|}$ followed by a
learned linear map to $\domain^\seqdim$, and the decoder is a learned
linear map followed by a binary decision rule.  Since neither map receives
a position index, the resulting architecture has no positional embeddings.

\begin{definition}[NoPE transformer]\label{def:nope-transformer}
    A \defn{NoPE transformer} is the algebraic transformer whose cores,
    recurrence and readout maps, wiring maps, encoder, and decoder are
    specified in \cref{eq:streaming-attention-interface,eq:nope-core-recurrence,eq:nope-core-readout,eq:nope-wiring}.
\end{definition}

\section{Available Semigroups by Attention Type}\label{app:prime-semigroups}

\begin{definition}\label{def:theUs}
    $\uone = (\{a,b\}, \cdot)$ with multiplication table 
    \begin{table}[H]
        \centering
        \begin{tabular}{c|cc}
            $\cdot$ & a & b \\
            \hline
            a & a & b \\
            b & a & b
        \end{tabular}
        \label{tab:uone-action}
    \end{table}   
    $\utwo = (\{0,1\}, \cdot)$ with multiplication table 
    \begin{table}[H]
        \centering
        \begin{tabular}{c|cc}
            $\cdot$ & 0 & 1 \\
            \hline
            0 & 0 & 0 \\
            1 & 0 & 1
        \end{tabular}
        \label{tab:utwo-action}
    \end{table}

    $\uthree = (\{a,b,1\}, \cdot)$ with multiplication table 
    \begin{table}[H]
        \centering
        \begin{tabular}{c|ccc}
            $\cdot$ & a & b & 1 \\
            \hline
            a & a & b & a \\
            b & a & b & b \\
            1 & a & b & 1
        \end{tabular}
        \label{tab:uthree-action}
    \end{table}
\end{definition} 
For each multiplication table of $U_i$, the entry in row $x$ and column $y$ is the product $xy$: first $x$, then $y$.
\windowOneDefinite*

\begin{proof}
    By \cref{def:streaming-window-one-attention}, reading an input
    $\insym$ replaces every query-indexed accumulator by
    $(\keyop(\insym),\valop(\insym))$, independently of the previous
    state.  Hence, every generator
    $\recfn^\head_\insym\in\transsem_\head$ is a constant
    transformation.  More generally, the transformation induced by a
    nonempty input word is the constant transformation determined by its
    final symbol.  Therefore, for all
    $s,t\in\transsem_\head,\ st=t$.
    In particular, this holds for every
    $t\in\transsem_\head^2$.  Taking $n=2$ in the definition of
    definiteness gives $\transsem_\head\in\Definite$.
\end{proof}

\windowOneContainsUone*

\begin{proof}
Take a singleton key set $\Keys=\{\key\}$ and choose three distinct values $\Values=\{\val_1,\val_a,\val_b\}$, so that $A=\{(\key,\val_1),(\key,\val_a),(\key,\val_b)\}$;
 write $\state_1,\state_a,\state_b$ for these three accumulator states.    
 Let the admissible inputs be $\insym_a$ and $\insym_b$, with values $\val_a$ and $\val_b$. 
By \cref{def:streaming-window-one-attention}, the induced transformations are the constant maps
    \begin{equation}
        \state \cdot F_a = \state_a,
        \qquad
        \state \cdot F_b = \state_b,
        \qquad \state \in A .
    \end{equation}
A product of constant maps is its last factor, so $F_xF_y=F_y$ for $x,y\in\{a,b\}$: four identities, each verified by evaluating both sides at any state. 
The set $A$ is invariant, and the surjections
    \begin{equation}
        \theta(\state_1)=1,\ \theta(\state_a)=a,\ \theta(\state_b)=b,
        \varphi(F_a)=a,\ \varphi(F_b)=b,
    \end{equation}
satisfy $\theta(\state\cdot F)=\theta(\state)\cdot\varphi(F)$ in all six cases, both sides being the label of the last-written pair. 
This shows the division $\barUone\divides(\transsem_\head,\states_\head)$.
\end{proof}
\sharpSoftRtrivial*

\begin{proof}
    Consider the deterministic semiautomaton induced by the head on
    $\states_\head$, and equip its states with the reachability preorder
    \begin{equation}
        \state\preceq\state'
        \quad\Longleftrightarrow\quad
        \state'=\state\cdot w
        \text{ for some }w\in\inset^*.
    \end{equation}
    This relation is reflexive and transitive.  We show that it is also
    antisymmetric.

    Let $d_\query(\state)$ be the denominator of the accumulator for
    query $\query$ in state $\state$.  Every sharp soft-attention
    transition satisfies
    \begin{equation}
        d_\query(\state)
        \leq
        d_\query(\state\cdot\insym)
        \qquad
        \text{for all }\query\in\Queries.
        \label{eq:sharp-denominator-order}
    \end{equation}
    Suppose $\state\preceq\state'$ and $\state'\preceq\state$.  Applying
    \cref{eq:sharp-denominator-order} along the two witnessing paths
    yields
    \begin{equation}
        d_\query(\state)
        \leq
        d_\query(\state')
        \leq
        d_\query(\state)
    \end{equation}
    for every query.  Hence no denominator advances anywhere along the
    path from $\state$ to $\state'$.  By the definition of sharp soft
    attention, whenever a denominator does not advance, its numerator
    also remains unchanged.  Every transition on this path is therefore
    a self-loop on the full head state, and $\state=\state'$.

    The reachability preorder is thus a partial order, so the head
    induces a partially ordered semiautomaton.  By the standard
    characterization of finite partially ordered semiautomata, their
    transition semigroups are $\Rrel$-trivial.  Therefore
    $\transsem_\head\in\Rtrivial$.
\end{proof}

\sharpSoftContainsUtwo*

\begin{proof}
    Use a scalar-valued head with one query and value zero.  Let $p$ be
    the significand precision, set $D=2^p$, and let
    $D^-=D-1$ be its immediate floating-point predecessor.  Choose an
    input $\insym_0$ whose exponential attention weight underflows to
    zero and an input $\insym_1$ with score zero, and hence exponential
    attention weight $1$.  Consider
    \begin{equation}
        \state_0=(0,D^-),
        \qquad
        \state_1=(0,D).
    \end{equation}
The input $\insym_0$ leaves both states unchanged and hence induces the identity on $\{\state_0,\state_1\}$.  Moreover,
    \begin{equation}
        \operatorname{fl}(D^-+1)=D,
        \qquad
        \operatorname{fl}(D+1)=D.
    \end{equation}
Thus, the sharp update induced by $\insym_1$ sends $\state_0$ to $\state_1$ and fixes $\state_1$: at $\state_0$ the denominator advances, whereas at $\state_1$ the increment is absorbed and the entire update is suppressed. 
Its restriction to $\{\state_0,\state_1\}$ is therefore the constant map $E$ with image $\state_1$, while $\insym_0$ restricts to the identity $I$ on the same pair of states.

The products are computed by evaluation: $EE=E$ and $IE=EI=E$, checked at $\state_0$ and at $\state_1$; this is the multiplication of \cref{tab:utwo-action}.
The set $\{\state_0,\state_1\}$ is invariant, and the surjections
\begin{equation}
    \theta(\state_0)=1,\ \theta(\state_1)=0,
    \qquad
    \varphi(I)=1,\ \varphi(E)=0,
\end{equation}
satisfy $\theta(\state\cdot F)=\theta(\state)\cdot\varphi(F)$ in all four cases; for instance $\theta(\state_0\cdot E)=\theta(\state_1)=0=1\cdot 0 =\theta(\state_0)\cdot\varphi(E)$. 
This is the division $\barUtwo\divides(\transsem_\head,\states_\head)$.

\end{proof}

\softAttentionAperiodic*

\begin{proof}
    Fix a query $\query$.  Upon reading a key--value pair $(\key,\val)$,
    the denominator is updated by
    \begin{equation}
        d_\query
        \longmapsto
        \operatorname{fl}\!\left(
            d_\query+\exp(\similarity(\query,\key))
        \right).
    \end{equation}
    The exponential is nonnegative, including when it underflows to zero.
    Since floating-point addition and rounding are order-preserving under
    the fixed semantics, the updated denominator is at least
    $d_\query$.  Thus every transition either strictly advances the
    denominator or leaves it unchanged through underflow or absorption.
    The same is true of every composition of transitions.  Consequently,
    under repeated application of any transformation
    $s\in\transsem_\head$, each denominator follows a nondecreasing
    sequence in a finite chain and therefore eventually stabilizes.

    It remains to rule out a cycle in the numerator after the denominator
    has stabilized.  For a fixed input and query, every scalar coordinate
    of the numerator is updated by addition of the fixed floating-point
    value
    \begin{equation}
        \operatorname{fl}\!\left(
            \exp(\similarity(\query,\key))\val
        \right).
    \end{equation}
    Addition by a fixed value and floating-point rounding are
    order-preserving.  Hence, on each scalar numerator coordinate, every
    $s\in\transsem_\head$ acts as an order-preserving self-map of a finite
    chain.  An order-preserving self-map of a finite chain has no
    nontrivial cycle: depending on whether $x\leq s(x)$ or
    $s(x)\leq x$, its orbit is respectively nondecreasing or
    nonincreasing and must stabilize.

    The head state has finitely many queries and finitely many scalar
    accumulator coordinates.  For each $s$, taking the maximum
    stabilization index over all coordinates and states gives
    $N_s\in\N$ such that $s^{N_s}=s^{N_s+1}$.  Since
    $\transsem_\head$ is finite, the maximum of the $N_s$ gives a common
    $N$ satisfying $s^N=s^{N+1}$ for every
    $s\in\transsem_\head$.  Therefore
    $\transsem_\head\in\aperiodic$.
\end{proof}

\softAttentionContainsUthree*
\begin{proof}
    It suffices to use one scalar-valued head with a singleton as the query set.
    Let $p$ be the significand precision of the floating-point format and
    choose the exactly representable values
    \begin{equation}
        C=2^{p+1},
        \qquad
        D=2^{p+2}.
    \end{equation}
    These values lie in the normal range of every standard IEEE binary
    format.  Their spacing implies
    \begin{equation}
        \operatorname{fl}(C+x)=C
        \quad\text{for }x\in\{-1,0,1\},
        \qquad
        \operatorname{fl}(D+1)=D.
        \label{eq:uthree-absorption}
    \end{equation}

    Take four inputs $\insym_0,\insym_C,\insym_{-C},\insym_1$ whose keys
    all have score zero against the unique query and whose respective
    values are $0,C,-C,1$.  Since $\exp(0)=1$, their numerator increments
    are precisely these four values, whereas each denominator increment
    is $1$.  Consider the three head states
    \begin{equation}
        \state_1=(-1,D),
        \qquad
        \state_a=(0,D),
        \qquad
        \state_b=(1,D),
    \end{equation}
    and set $\states'=\{\state_1,\state_a,\state_b\}$.  By
    \cref{eq:uthree-absorption}, the transformation induced by
    $\insym_0$ fixes every element of $\states'$ and therefore restricts
    to $\identity{\states'}$.

    Let $F_a$ be the transformation induced by the word
    $\insym_C\insym_{-C}$.  For every $(n,D)\in\states'$,
    \begin{equation}
        (n,D)
        \xmapsto{\insym_C}
        (C,D)
        \xmapsto{\insym_{-C}}
        (0,D)
        =\state_a.
    \end{equation}
    Thus, $F_a$ restricts to the constant transformation with image
    $\state_a$.  Similarly, the word
    $\insym_C\insym_{-C}\insym_1$ induces a transformation $F_b$ such
    that
    \begin{equation}
        (n,D)
        \xmapsto{\insym_C\insym_{-C}}
        (0,D)
        \xmapsto{\insym_1}
        (1,D)
        =\state_b
    \end{equation}
    for every $(n,D)\in\states'$.  Hence, $F_b$ restricts to the constant
    transformation with image $\state_b$.
The set $\states'$ is invariant under $\identity{\states'}$, $F_a$, and $F_b$, and a product of constant maps is its last factor, so on $\states'$
    \begin{equation}
        F_xF_y=F_y,
        \quad
        \identity{\states'}F_x=F_x\identity{\states'}=F_x,
        \quad x,y\in\{a,b\},
    \end{equation}
each identity verified by evaluating at any state. This is the multiplication in \cref{tab:uthree-action}. The surjections
\begin{align}
    \theta(\state_1)=1,\quad \theta(\state_a)&=a,\quad \theta(\state_b)=b,\\
    \varphi(\identity{\states'})=1,\quad \varphi(F_a)&=a,\quad \varphi(F_b)=b,
\end{align}
satisfy $\theta(\state\cdot F)=\theta(\state)\cdot\varphi(F)$ in all nine cases, both sides being the label of the last reset, or the old label if none.
This is the division $\barUthree\divides(\transsem_\head,\states_\head)$.
\end{proof}

\begin{remark}[Equality against division]\label{rem:equality-vs-division}
Here, they hold only on $\states'$: before saturation, the denominator counts letters, one per input, and a length counter separates words of different lengths.
For instance $(0,0)\cdot F_b=(1,3)$ while $(0,0)\cdot F_aF_b=(1,5)$, so $F_aF_b\neq F_b$ in $\transsem_\head$, and $F_aF_a \neq F_a$ although $a$ is idempotent in $\uthree$.
The subsemigroup $\gensem{\identity{}, F_a, F_b}$ is therefore strictly larger than $\uthree$, of which it sees the relations only after restriction to the saturated slice: a quotient of a subobject, which is precisely a division, and the reason \cref{lem:sharp-soft-contains-utwo,lem:soft-attention-contains-uthree} assert $\divides$ where \cref{lem:window-one-contains-uone} could assert an isomorphism onto a subobject.
\end{remark}

\section{Semigroup Decomposition Beyond Krohn--Rhodes}\label{app:subregular-kr}
Krohn--Rhodes theory \citep{krohn-rhodes-1965-algebraic} describes the decomposition of any aperiodic transformation semigroup into a wreath product of transformation semigroups whose semigroup is $\uthree$. However, transformers may not be able to instantiate $\uthree$, meaning none of the basic constituents of such a decomposition is attainable. Therefore, we utilize a refinement of Krohn--Rhodes due to \citet{STIFFLER1973159} that further decomposes aperiodic semigroups into smaller semigroups $\uone,\utwo \subsem \uthree$, which we will show can be realized by different types of attention heads.

\begin{theorem}[\citet{krohn-rhodes-1965-algebraic}, semigroup version]
    Every finite transformation semigroup $(\semigroup,\states)$ divides a finite wreath product with factors alternating between copies of $\barUthree$ and simple group factors.
\end{theorem}

\begin{corollary}
\label{thm:aperiodic-decomposition}
    Let $(\semigroup,\states)$ be a finite transformation semigroup. Then the semigroup $\semigroup$ is aperiodic, $\semigroup \in \aperiodic$, if and only if $(\semigroup,\states)$ divides a finite wreath product of copies of $\barUthree$:
    \begin{equation}
        \semigroup\in\Aperiodic
        \quad\Longleftrightarrow\quad
        \exists\,\states\colon
        (\semigroup,\states)
        \divides
        \barUthree\wr\cdots\wr\barUthree.
    \end{equation}
\end{corollary}

\begin{theorem}[\citet{STIFFLER1973159}, Thm 3.4 (a) and Fact 1.9]
\label{thm:definite-decomposition}
    For every finite semigroup $\semigroup$,
    \begin{equation}
        \semigroup\in\Definite
        \quad\Longleftrightarrow\quad
        \exists\,\states\colon
        (\semigroup,\states)
        \divides
        \barUone\wr\cdots\wr\barUone.
    \end{equation}
    Moreover, $\Definite$ is closed under finite wreath products and
    division. Consequently, if a finite transformation semigroup
    $(\semigroup,\states)$ divides a finite cascade whose component
    semigroups are definite, then
    $\semigroup\in\Definite$.
\end{theorem}

\begin{theorem}[\citet{STIFFLER1973159}, Thm. 3.4(b) and Fact 1.9]
\label{thm:rtrivial-decomposition}
    For every finite semigroup $\semigroup$,
    \begin{equation}
        \semigroup\in\Rtrivial
        \quad\Longleftrightarrow\quad
        \exists\,\states\colon
        (\semigroup,\states)
        \divides
        \barUtwo\wr\cdots\wr\barUtwo.
    \end{equation}
    Moreover, $\Rtrivial$ is closed under finite wreath products and
    division. Consequently, if a finite transformation semigroup
    $(\semigroup,\states)$ divides a finite cascade whose component
    semigroups are $\mathcal R$-trivial, then
    $\semigroup\in\Rtrivial$.
\end{theorem}

\begin{theorem}[\citet{STIFFLER1973159}, Thm 3.4 (c), Fact 1.9, Cor. 12 (c) and Rem. 2.13]
\label{thm:locally-r-decomposition}
    For every finite semigroup $\semigroup$,
    \begin{equation}
        \semigroup\in\LocallyR
        \quad\Longleftrightarrow\quad
        \exists\,\states\colon
        (\semigroup,\states)
        \divides
        \barUone\wr\dots\wr\barUone\wr\barUtwo\wr\cdots\wr\barUtwo.
    \end{equation}
    Furthermore, $\Definite$ and $\Rtrivial$ are closed under finite wreath products and
    division. Therefore, if a finite transformation semigroup
    $(\semigroup,\states)$ divides a finite cascade whose component
    semigroups are $\mathcal R$-trivial and definite, then $\semigroup\in\LocallyR$.
\end{theorem}

\begin{remark}
    Note the different notational conventions: 
    \begin{enumerate}
        \item Whereas \citet{STIFFLER1973159} uses right wreath products, we define wreath factors to act from the left to better match the common indexing of layers in neural networks. Hence, our restatement of \cref{thm:locally-r-decomposition} has $\barUone$s on the left (think: closer to the transformer input) and $\barUtwo$s on the right (higher up in the layers), the reverse of \citet{STIFFLER1973159}.
        \item We adopt the notation of \citet{STIFFLER1973159} for the prime semigroups $\uone,\utwo, \uthree$. Note that these have different names elsewhere; for instance, our $\uthree$ is called $\utwo$ in \citet{eilenberg1976}, etc.
    \end{enumerate}
\end{remark}

\end{document}